\documentclass[letterpaper,conference]{IEEEtran}
\usepackage[T1]{fontenc}
\usepackage[utf8]{inputenc}
\usepackage{algorithm2e}
\usepackage{amsmath}
\usepackage{amssymb}
\usepackage{graphicx}

\makeatletter

\usepackage{amsfonts}
\usepackage{amsthm}
\usepackage{algorithmic}
\usepackage{url}
\usepackage{ifthen}
\usepackage{cite}
\theoremstyle{definition}
\newtheorem{definition}{Definition}[section]
\theoremstyle{lemma}
\newtheorem{theorem}{Theorem}
\newtheorem{lemma}[theorem]{Lemma}\newtheorem{proposition}{Proposition}
\DeclareMathOperator*{\argmin}{arg\,min}

\newcommand{\norm}{\text{Norm}}

\newcommand{\trace}{\text{Trace}}

\makeatother

\begin{document}
\title{On the Optimality of Gauss's Algorithm over Euclidean Imaginary Quadratic
Fields}
\author{ \IEEEauthorblockN{Christian Porter } \IEEEauthorblockA{Department of EEE\\
 Imperial College London\\
 London, SW7 2AZ, United Kingdom \\
 Email: c.porter17@imperial.ac.uk} \and \IEEEauthorblockN{Shanxiang Lyu} \IEEEauthorblockA{College of Cyber Security\\
 Jinan University\\
 Guangzhou 510632, China\\
 Email: s.lyu14@imperial.ac.uk} \and \IEEEauthorblockN{Cong Ling} \IEEEauthorblockA{Department of EEE\\
 Imperial College London\\
 London, SW7 2AZ, United Kingdom \\
 Email: c.ling@imperial.ac.uk}}
\maketitle
\begin{abstract}
In this paper, we continue our previous work on the reduction of algebraic
lattices over imaginary quadratic fields for the special case when
the lattice is spanned over a two dimensional basis. In particular,
we show that the algebraic variant of Gauss's algorithm returns a
basis that corresponds to the successive minima of the lattice in
polynomial time if the chosen ring is Euclidean. 
\end{abstract}

\section{Introduction}

Lattice reduction algorithms over algebraic number fields have attracted
great attention in recent years. In network information theory, efficient
techniques for finding the network coding matrices in compute-and-forward
boil down to designing a lattice reduction algorithm where the direct-sums
are defined by the space of codes \cite{Tunali2015,DBLP:journals/tit/HuangNW18,Vazquez-CastroO14}.
In cryptography, analyzing the bit-level security of ideal lattice
based NTRU or fully homomorphic encryption schemes through a sub-field
algorithm has been shown effective \cite{Kim2017,AlbrechtBD16}.

Since the Lenstra-Lenstra-Lovász (LLL) algorithm is one of the most
celebrated lattice reduction algorithms to date, its extension from
over real/rational space to higher dimensional space has been studied
extensively. Napias first generalised the LLL algorithm to lattices
spanned over imaginary quadratic rings and certain quaternion fields
\cite{Napias1996}; later Fieker, Pohst and Stehle investigated more
fundamental properties of algebraic lattices \cite{Fieker1996,Fieker2010}.
Recently Kim and Lee proposed an efficient LLL algorithm over bi-quadratic
field whose quantization step requires a Euclidean domain \cite{Kim2017},
while our work on LLL over imaginary quadratic fields showed that
a Euclidean domain is needed to make the algorithm convergent \cite{DBLP:lyu2018-ALLL}. 

As a special case of the LLL algorithm for (real/rational) lattices
over two dimensions, conventional Gauss's algorithm has been proved
to return the two vectors corresponding to the successive minima of
the lattice \cite{Micciancio2002}. However, the algebraic analog
of Gauss's algorithm, has not, so far, been analyzed. It remains unknown
whether the algebraic Gauss's algorithm possesses the optimality properties
as its counter-part.

To address this issue, we take a modest step to investigate Gauss's
algorithm over Imaginary Quadratic Fields. When the ring of integers
is a Euclidean domain, we prove that Gauss's algorithm returns a basis
corresponding to the successive minima of an algebraic lattice. This
result is further explained through numerical examples. Specifically,
we show how the algorithm finds the two successive minima when the
domain is Euclidean, and how the algorithm fails to work when it is
non-Euclidean.

\section{Preliminaries}

We begin by defining some familiar concepts in algebraic number theory
and lattice theory. Let $K$ be a complex quadratic extension of $\mathbb{Q}$,
i.e. $K=\mathbb{Q}(\sqrt{-d})$ for some positive square-free integer
$d$ that is not equal to $1$. Then recall the ring of integers of
$K$ (maximal order), $\mathcal{O}_{K}$, is $\mathbb{Z}[\xi]$, where
\[
\xi=\begin{cases}
\sqrt{-d} & \quad\text{if}\hspace{0.5cm}-d\equiv2,3\mod4,\\
\frac{1+\sqrt{-d}}{2} & \quad\text{if}\hspace{0.5cm}-d\equiv1\mod4.
\end{cases}
\]
\begin{definition} A field $K$ is said to be \emph{norm-Euclidean}
if, for all $x\in K$, there exists $q\in\mathcal{O}_{K}$ its ring
of integers such that 
\[
|\norm_{K/\mathbb{Q}}(x-q)|<1,
\]
where $\norm_{K/\mathbb{Q}}$ denotes the algebraic norm of $K$.
We denote the value $\mathcal{M}(K):=\max_{x\in K}\min_{q\in\mathcal{O}_{K}}|\norm_{K/\mathbb{Q}}(x-q)|$
the \emph{Euclidean minimum} of $K$. \end{definition}

\begin{proposition} Let $K=\mathbb{Q}(\sqrt{-d})$ be an imaginary
quadratic field with ring of integers $\mathcal{O}_{K}$. Then we
have 
\[
\mathcal{M}(K)=\begin{cases}
\frac{1+d}{4}\hspace{0.2cm}\text{if}\hspace{0.1cm}-d\equiv2,3\mod4,\\
\frac{(1+d)^{2}}{16d}\hspace{0.2cm}\text{if}\hspace{0.1cm}-d\equiv1\mod4.
\end{cases}
\]

Hence, $K$ is norm-Euclidean if and only if $d\in\{1,2,3,7,11\}$.
\end{proposition}

\begin{proof} See \cite{DBLP:lyu2018-ALLL}. \end{proof} For imaginary
quadratic fields, we may analytically extend the norm function to
all complex numbers using the absolute value. Moreover, we have $\max_{x\in K}\min_{q\in\mathcal{O}_{K}}|\norm_{K/\mathbb{Q}}(x-q)|=\max_{x\in\mathbb{C}}\min_{q\in\mathcal{O}_{K}}|x-q|^{2}$
as the maximum distance with respect to the absolute value is achieved
at a rational point. We say that $x\in\mathbb{C}$ is fully $\mathbb{Z}[\xi]$-reduced
if $|x|\leq|x-q|$ for all $q\in\mathbb{Z}[\xi]$.

\begin{lemma} Let $x\in\mathbb{C}$ be fully $\mathbb{Z}[\xi]$-reduced.
Then $|\Re(x)|\leq1/2,|\Im(x)|\leq\sqrt{d}/2$ if $\xi=\sqrt{-d}$
or $|\Re(x)|\leq1/2$, $|\Im(x)|\leq\frac{1}{\sqrt{d}}\left(-|\Re(x)|+\frac{1+d}{4}\right)$
if $\xi=\frac{1+\sqrt{-d}}{2}$. \end{lemma}

\begin{proof} Define the map $\phi(x+iy)=(x,y)$ for all $x+iy\in\mathbb{C}$.
Then $|x+iy|=\|(x,y)\|$. When $-d\equiv2,3\mod4$, $\mathbb{Z}[\xi]$
generates the lattice with basis $(1,0),(0,\sqrt{d})$, otherwise
$\mathbb{Z}[\xi]$ generates the lattice with basis $(1,0),(1/2,\sqrt{d}/2)$.
The bounds that form the fundamental region of these lattices correspond
to the bounds given in the proposition (see fig. 1 for reference).
\end{proof} 
\begin{figure}
\centering \includegraphics[width=0.8\columnwidth]{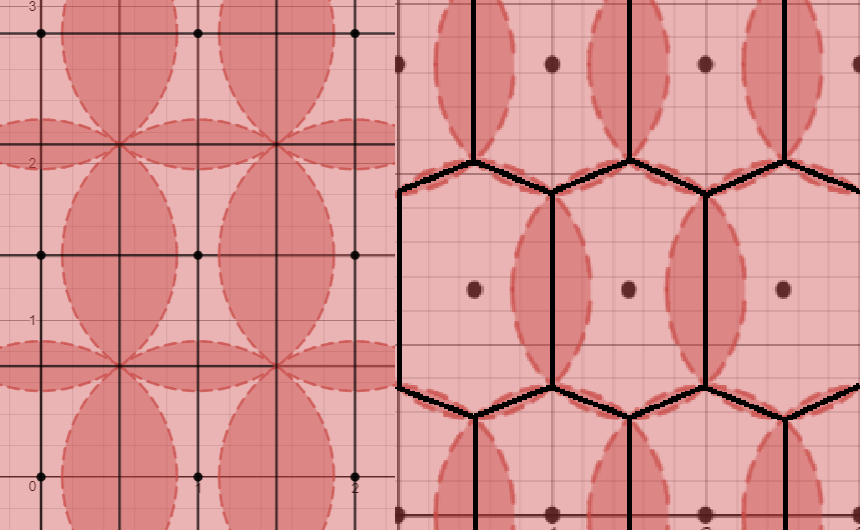}
\caption{Left: lattice generated by $\mathbb{Z}[\sqrt{-2}]$, tesselated by
rectangles, Right: lattice generated by $\mathbb{Z}[\frac{1+\sqrt{-7}}{2}]$,
tesselated by "stretched" hexagons.}
\label{fig:my_label} 
\end{figure}
A $\mathbb{Z}[\xi]$-module $\Lambda$ is an abelian group with a
binary operation over $\mathbb{Z}[\xi]$ that satisfies the axioms
of scalar multiplication. In general, modules need not have a basis,
and those that are are denoted free modules. A subset of $\Lambda$
forms a basis for $\Lambda$ if the basis is linearly independent
over $\mathbb{Z}[\xi]$ and also spans $\Lambda$ over $\mathbb{Z}[\xi]$.
We denote discrete free $\mathbb{Z}[\xi]$-submodules of $\mathbb{C}^{n}$
\emph{algebraic lattices}. Algebraic lattices can be expressed by
a basis $\mathbf{b}_{1},\dots,\mathbf{b}_{n}$ whose linear sum with
scalar multiplication over $\mathbb{Z}[\xi]$ span $\Lambda$.

\begin{definition} The $j$th successive minima of an algebraic lattice
is the smallest such number $\lambda_{j}$ such that the ball of radius
$\lambda_{j}$ (under an appropriate norm) contains $j$ linearly
independent lattice vectors over $\mathbb{Z}[\xi]$. \end{definition} 

\subsection{Classical Gauss's algorithm}

Aside from Euclid's famous greatest divisor algorithm, Gauss's lattice
reduction algorithm is one of the first examples of a lattice reduction
algorithm. Gauss defined the notion of a reduced basis over two dimensions
as the following. \begin{definition} An ordered basis $\{\mathbf{b}_{1},\mathbf{b}_{2}\}\in\mathbb{R}^{n}$
of a two dimensional lattice is reduced if $\|\mathbf{b}_{1}\|\leq\|\mathbf{b}_{2}\|\leq\|\mathbf{b}_{2}+p\mathbf{b}_{1}\|$
for all $p\in\mathbb{Z}$. \end{definition} The following algorithm
returns a reduced basis in the sense of Gauss. \IncMargin{1em}
\begin{algorithm}
\SetKwInOut{Input}{input}\SetKwInOut{Output}{output} \Input{An
ordered basis $\{\mathbf{b}_{1},\mathbf{b}_{2}\}\in\mathbb{R}^{n}$
of a two dimensional lattice spanned over $\mathbb{Z}$.} \Output{A
Gauss reduced basis.} \BlankLine \While{$\|\mathbf{b}_{1}\|<\|\mathbf{b}_{2}\|$}{$\mu_{12}=\langle\mathbf{b}_{1},\mathbf{b}_{2}\rangle/\|\mathbf{b}_{1}\|^{2}$\;
$\mathbf{b}_{2}=\mathbf{b}_{2}-\lfloor\mu_{12}\rceil\mathbf{b}_{1}$
\; swap $\mathbf{b_{1}},\mathbf{b_{2}}$} 
\end{algorithm}
\DecMargin{1em} Notice we get a basis whose Gram-Schmidt coefficients
round to zero, as such $|\mu_{12}|,|\mu_{21}|\leq1/2$. Then taking
an arbitrary vector in the lattice $\mathbf{v}=x\mathbf{b}_{1}+y\mathbf{b}_{2}$
where $x,y\in\mathbb{Z}$, we get 
\begin{align*}
\|\mathbf{v}\||^{2} & =x^{2}\|\mathbf{b}_{1}|^{2}+2xy\langle\mathbf{b}_{1},\mathbf{b}_{2}\rangle+y^{2}\|\mathbf{b}_{2}\|^{2}\\
 & \geq x^{2}\|\mathbf{b}_{1}\|^{2}-|xy|\|\mathbf{b}_{1}\|^{2}+y^{2}\|\mathbf{b}_{1}\|^{2}\\
 & =(x-y)^{2}\|\mathbf{b}_{1}\|^{2}+|xy|\|\mathbf{b}_{1}\|^{2}.
\end{align*}
If $x$ is nonzero, since $(x-y)^{2},|xy|\in\mathbb{N}$ and at least
one must be nonzero, $\mathbf{b}_{1}$ must be the shortest vector.
If $y$ is nonzero, letting $f(x,y)=(x-y)^{2}+|xy|$ we have 
\begin{align*}
\|\mathbf{v}\|^{2} & =y^{2}(\|\mathbf{b}_{2}\|^{2}-\|\mathbf{b}_{1}\|^{2})\\
 & +(x^{2}+y^{2})\|\mathbf{b}_{1}\|^{2}+2xy\langle\mathbf{b}_{1},\mathbf{b}_{2}\rangle\\
 & \geq y^{2}(\|\mathbf{b}_{2}\|^{2}-\|\mathbf{b}_{1}\|^{2})+f(x,y)\|\mathbf{b}_{1}\|^{2}\\
 & \geq(y^{2}-1)(\|\mathbf{b}_{2}\|^{2}-\|\mathbf{b}_{1}\|^{2})+\|\mathbf{b}_{2}\|^{2},
\end{align*}
and since $y$ is nonzero, $y^{2}-1\geq0$ so $\mathbf{b}_{2}$ corresponds
to the second successive minima for the lattice.

\section{ Algebraic Lattice reduction in two dimensions}

For our work, we use the complex Euclidean ($l_{2}$) norm to measure
the length of lattice vectors and the regular complex inner product.
Unlike algebraic lattices spanned over other rings, we do not need
to embed the ring structure before measuring the norm, as the Euclidean
norm already takes the complex conjugate into account. Define the
quantisation function $q_{K}:\mathbb{C}\to\mathbb{Z}[\xi]$ such that
$q_{K}(x)=\argmin_{\mu\in\mathbb{Z}[\xi]}|x-\mu|$. A specific definition
of how the quantisation function works can be found in \cite{DBLP:lyu2018-ALLL}.

Lagrange and Gauss have given the reduction criteria for a two dimensional
real basis. We first generalize this criteria to over complex quadratic
rings. \begin{definition} A basis $\mathbf{b}_{1},\mathbf{b}_{2}\in\mathbb{C}^{n}$
is Gauss reduced if $\|\mathbf{b}_{1}\|\leq\|\mathbf{b}_{2}\|\leq\|\mathbf{b}_{2}+p\mathbf{b}_{1}\|$
for all $p\in\mathbb{Z}[\xi]$. \end{definition} The following algorithm,
which is a special case of algebraic LLL in two dimensions, computes
a Gauss reduced basis. \\
 \IncMargin{1em} 
\begin{algorithm}
\SetKwInOut{Input}{input}\SetKwInOut{Output}{output} \Input{An
ordered basis $\{\mathbf{b}_{1},\mathbf{b}_{2}\}\in\mathbb{C}^{n}$
of a two dimensional algebraic lattice $\Lambda$ and a relevant ring
$\mathbb{Z}[\xi]$ that we want to reduce the basis over.} \Output{A
Gauss reduced basis.} \BlankLine \While{$\|\mathbf{b}_{1}\|<\|\mathbf{b}_{2}\|$}{$\mu_{12}=\langle\mathbf{b}_{1},\mathbf{b}_{2}\rangle/\|\mathbf{b}_{1}\|^{2}$\;
$\mathbf{b}_{2}=\mathbf{b}_{2}-q_{K}(\mu_{12})\mathbf{b}_{1}$ \;
swap $\mathbf{b_{1}},\mathbf{b_{2}}$} 
\end{algorithm}
\DecMargin{1em} \begin{theorem} Let $\mathbf{b}_{1},\mathbf{b}_{2}$
be an output basis of the algorithm above. Then $\|\mathbf{b}_{1}\|=\lambda_{1},\|\mathbf{b}_{2}\|=\lambda_{2}$
if $\mathbb{Z}[\xi]$ is the ring of integers of a norm-Euclidean
domain (i.e., $d=1,2,3,7,11$). \end{theorem} \begin{proof}

We first show that the Gram-Schmidt coefficients are fully $\mathbb{Z}[\xi]$-reduced,
i.e. the GS coefficients of the output basis are rounded to zero.
Let $\mathbf{b}_{2}$ be the input vector in the last run of the algorithm
before $\mathbf{b}_{2},\mathbf{b}_{1}$ are output as the reduced
basis, and let $\mu_{12}$ be the GS coefficient between $\mathbf{b}_{1},\mathbf{b}_{2}$.
Then $q_{K}(\langle\mathbf{b}_{1},\mathbf{b}_{2}\rangle/\|\mathbf{b}_{1}\|^{2})=q_{K}(\frac{1}{\|\mathbf{b}_{1}\|^{2}}(\langle\mathbf{b}_{1},\mathbf{b}_{2}\rangle-q_{K}(\mu_{12})\|\mathbf{b}_{1}\|^{2}))=q_{K}(\epsilon)$,
where $\epsilon=\mu_{12}-q_{K}(\mu_{12})$ has already been fully
reduced, by the definition of the quantisation function. Since no
swap has occurred (since the basis has been output), $\|\mathbf{b}_{1}\|\leq\|\mathbf{b}_{2}\|$
so the same argument follows for the GS coefficient between $\mathbf{b}_{2},\mathbf{b}_{1}$.

To prove that $\|\mathbf{b}_{1}\|=\lambda_{1}$, we denote an arbitrary
lattice vector $\mathbf{v}=p_{1}\mathbf{b}_{1}+p_{2}\mathbf{b}_{2}$
where $p_{1},p_{2}\in\mathbb{Z}[\xi]$, and analyze its norm function:
\begin{equation}
\|\mathbf{v}\|^{2}=|p_{1}|^{2}\|\mathbf{b}_{1}\|^{2}+|p_{2}|^{2}\|\mathbf{b}_{2}\|^{2}+2\Re(\overline{p_{1}}p_{2}\langle\mathbf{b}_{1},\mathbf{b}_{2}\rangle).\label{prove1_v1}
\end{equation}

We examine the cases $-d\equiv1,2\mod4$ and $-d\equiv3\mod4$ separately.
When the chosen ring is in the form of $\xi=\sqrt{-d}$, we let $p_{1}=x+y\sqrt{-d},p_{2}=z+w\sqrt{-d}$
where $x,y,z,w\in\mathbb{Z}$. Then $\overline{p_{1}}p_{2}=(xz+dyw)+\sqrt{-d}(xw-yz)$,
and 
\begin{align*}
2\Re(\overline{p_{1}}p_{2}\langle\mathbf{b}_{1},\mathbf{b}_{2}\rangle) & =2(xz+dyw)\Re(\langle\mathbf{b}_{1},\mathbf{b}_{2}\rangle)\\
 & -2\sqrt{d}(xw-yz)\Im(\langle\mathbf{b}_{1},\mathbf{b}_{2}\rangle).
\end{align*}
Since the GS coefficients are fully reduced, we have: 
\[
\begin{cases}
2(xz+dyw)\Re(\langle\mathbf{b}_{1},\mathbf{b}_{2}\rangle)\geq-|xz+dyw|\|\mathbf{b}_{1}\|^{2},\\
-2\sqrt{d}(xw-yz)\Im(\langle\mathbf{b}_{1},\mathbf{b}_{2}\rangle)\geq-d|xw-yz|\|\mathbf{b}_{1}\|^{2}.
\end{cases}
\]
Based on this, the r.h.s. of Eq. (\ref{prove1_v1}) can be lower bounded:
\begin{align}
\|\mathbf{v}\|^{2}\geq Q_{1}'(x,y,z,w)\|\mathbf{b}_{1}\|^{2},\label{eq:lambda_1_case1}
\end{align}
where 
\begin{align*}
Q_{1}'(x,y,z,w) & \triangleq(x^{2}+dy^{2}+z^{2}+dw^{2}\\
 & -|xz+dyw|-d|xw-yz|).
\end{align*}
Letting $Q_1(x,y,z,w)\triangleq(x^{2}+dy^{2}+z^{2}+dw^{2} -(xz+dyw)-d(xw-yz))$, we note that the codomain of $Q_1'$ is a subset of the codomain of $Q_1$ (this can be seen by changing the signs of $x,y,z,w$ around until the functions are equivalent), showing positive-definiteness of $Q_1$ immediately yields that $Q_1'$ is also positive-definite. The 4-D symmetric matrix w.r.t. quadratic form $Q_{1}(x,y,z,w)$ can
be written as 
\[
\mathbf{Q}_{1}=\left[\begin{array}{cccc}
1 & 0 & -\frac{1}{2} & -\frac{d}{2}\\
0 & d & \frac{d}{2} & -\frac{d}{2}\\
-\frac{1}{2} & \frac{d}{2} & 1 & 0\\
-\frac{d}{2} & -\frac{d}{2} & 0 & d
\end{array}\right].
\]
The four eigenvalues of $\mathbf{Q}_{1}$ are: 
\[
\begin{cases}
\frac{d-\sqrt{5d^{2}-6d+9}+3}{4},\\
\frac{d+\sqrt{5d^{2}-6d+9}+3}{4},\\
\frac{3d-\sqrt{13d^{2}-6d+1}+1}{4},\\
\frac{3d+\sqrt{13d^{2}-6d+1}+1}{4}.
\end{cases}
\]
We therefore conclude that $\mathbf{Q}_{1}$ has four positive eigenvalues
and hence being positive definite with only $d=1,2$ in this case.
Along with $Q(x,y,z,w)\in\mathbb{Z}$, we arrive at $\|\mathbf{v}\|^{2}\geq\|\mathbf{b}_{1}\|^{2}$
when $d=1,2$.

When the chosen ring is in the form of $\xi=\frac{1+\sqrt{-d}}{2}$,
like before, letting $p_{1}=x+y\frac{1+\sqrt{-d}}{2},p_{2}=z+w\frac{1+\sqrt{-d}}{2}$,
we have $\overline{p_{1}}p_{2}=(xz+1/2(yz+xw)+\frac{1+d}{4}yw)+(\sqrt{-d}/2)(xw-yz)$.
Then 
\begin{align*}
2\Re(\overline{p_{1}}p_{2}\langle\mathbf{b}_{1},\mathbf{b}_{2}\rangle)= & 2(xz+1/2(yz+xw)\\& +\frac{1+d}{4}yw)\Re(\langle\mathbf{b}_{1},\mathbf{b}_{2}\rangle) \\ &-\sqrt{d}(xw-yz)\Im(\langle\mathbf{b}_{1},\mathbf{b}_{2}\rangle).
\end{align*}
Using the following inequality from the ``fully-reduced'' constraints:
$$
|\Im(x)|\leq\frac{1}{\sqrt{d}}\left(-|\Re(x)|+\frac{1+d}{4}\right),
$$
similarly to before, we obtain the inequality
\begin{align*}
    \|\mathbf{v}\|^2 &\geq (x^2 + xy +\frac{1+d}{4}y^2 + z^2 + zw + \frac{1+d}{4}w^2  \\ &- \frac{1+d}{4}|xw-yz|)\|\mathbf{b}_1\|^2 -|\Re(\langle \mathbf{b}_1,\mathbf{b}_2 \rangle)|\\ & *(|2xz + \frac{1+d}{2}yw + xw +yz| - |xw-yz|).
\end{align*}
Focusing on the term $(|2xz + \frac{1+d}{2}yw + xw +yz| - |xw-yz|)$, we note that one of the $xw,yz$ on the left hand term must annihilate with one on the right hand term, and one must sum to two times the variable (the choice of which does not matter for our case, as the overall function is symmetric in $xw,yz$). We choose $xw$ to annihilate and $yz$ to coalesce. Then clearly, all terms whose coefficient is $|\Re(\langle \mathbf{b}_1,\mathbf{b}_2 \rangle)|$ are negative, so the minimum is achieved at $|\Re(\langle \mathbf{b}_1,\mathbf{b}_2 \rangle)|=1/2 \|\mathbf{b}_1\|^2$. Once again, to show the above is greater than or equal to $\|\mathbf{b}_1\|^2$ for all $x,y,z,w$, we construct a ``larger'' quadratic form, $ Q_{2}(x,y,z,w)$, and show its positive-definiteness, where:
\begin{align*}
Q_{2}(x,y,z,w) & \triangleq(x^{2}+ xy +\frac{1+d}{4}y^{2}+z^{2}+zw\\ &+\frac{1+d}{4}w^{2}-\frac{1+d}{4}xw+\left(\frac{1+d}{4}-1\right)yz \\ &-xz -\frac{1+d}{4}yw.
\end{align*}
The symmetric matrix w.r.t. the quadratic form $Q_{2}(x,y,z,w)$ and its
corresponding eigenvalues are respectively: 
\[
\mathbf{Q}_{2}=\left[\begin{array}{cccc}
1 & 1/2 & -\frac{1}{2} & -\frac{1+d}{8}\\
1/2 & \frac{1+d}{4} & \frac{1}{2}\left(\frac{1+d}{4}-1\right) & -\frac{1+d}{8}\\
-\frac{1}{2} & \frac{1}{2}\left(\frac{1+d}{4}-1\right) & 1 & 1/2\\
-\frac{1+d}{8} & -\frac{1+d}{8} & 1/2 & \frac{1+d}{4}
\end{array}\right],
\]
\[
\begin{cases}
\frac{2D+2-\sqrt{9D^{2}-10D^{3}+10-4\frac{D^3-D^2+2}{\sqrt{D^2-2D+2}}}-\sqrt{D^2-2D+2}}{4},\\
\frac{2D+2+\sqrt{9D^{2}-10D^{3}+10-4\frac{D^3-D^2+2}{\sqrt{D^2-2D+2}}}-\sqrt{D^2-2D+2}}{4},\\
\frac{2D+2-\sqrt{9D^{2}-10D^{3}+10+4\frac{D^3-D^2+2}{\sqrt{D^2-2D+2}}}+\sqrt{D^2-2D+2}}{4},\\
\frac{2D+2+\sqrt{9D^{2}-10D^{3}+10+4\frac{D^3-D^2+2}{\sqrt{D^2-2D+2}}}+\sqrt{D^2-2D+2}}{4},
\end{cases}
\]
where $D=\frac{1+d}{4}$. Through checking the eigenvalues, it shows that $\mathbf{Q}_{2}$ is
positive definite when $d=3,7,11$; therefore $\|\mathbf{v}\|^{2}\geq\|\mathbf{b}_{1}\|^{2}$
is reached.

To prove that $\|\mathbf{b}_{2}\|=\lambda_{2}$, we leverage the technique
in \cite{Yao2002}. For both cases of $\xi$, we construct a vector
$p_{1}\mathbf{b}_{1}+p_{2}\mathbf{b}_{2}$ with $p_{1},p_{2}\in\mathbb{Z}[\xi]$,
$p_{2}\neq0$. When the chosen ring is in the form of $\xi=\sqrt{-d}$,
we have 
\begin{align*}
 & \|p_{1}\mathbf{b}_{1}+p_{2}\mathbf{b}_{2}\|^{2}=|p_{2}|^{2}(\|\mathbf{b}_{2}\|^{2}-\|\mathbf{b}_{1}\|^{2})\\
 & +(|p_{1}|^{2}+|p_{2}|^{2})\|\mathbf{b}_{1}\|^{2}+2\Re(\overline{p_{1}}p_{2}\langle\mathbf{b}_{1},\mathbf{b}_{2}\rangle)\\
 & \geq|p_{2}|^{2}(\|\mathbf{b}_{2}\|^{2}-\|\mathbf{b}_{1}\|^{2})+Q_{1}(x,y,z,w)\|\mathbf{b}_{1}\|^{2}\\
 & \geq(|p_{2}|^{2}-1)(\|\mathbf{b}_{2}\|^{2}-\|\mathbf{b}_{1}\|^{2})+\|\mathbf{b}_{2}\|^{2}\\
 & \geq\|\mathbf{b}_{2}\|^{2}.
\end{align*}
This shows $\mathbf{b}_{2}$ is the shortest lattice vector that is
independent of $\mathbf{b}_{1}$. The proof for the case $\xi=\frac{1+\sqrt{-d}}{2}$
follows the same way by replacing $Q_{1}(x,y,z,w)$ with $Q_{2}(x,y,z,w)$.
\end{proof}

\section{Numerical examples}

\textbf{Example 1 (Euclidean domain).} Consider the field $K=\mathbb{Q}\left(\sqrt{-3}\right)$
and its maximal ring of integers $\mathbb{Z}[\omega]$. Suppose the
input lattice basis is 
\[
\mathbf{B}=\left[\begin{array}{cc}
4+\omega & 1+4\omega\\
-1+5\omega & 1+2\omega
\end{array}\right].
\]
The algebraic reduction on this basis will consist of a swap, a size
reduction, and another swap, to yield the reduced basis 
\[
\tilde{\mathbf{B}}=\left[\begin{array}{cc}
-3+3\omega & 1+4\omega\\
2-3\omega & 1+2\omega
\end{array}\right],
\]
which satisfies $\left\Vert \tilde{\mathbf{b}}_{1}\right\Vert ^{2}=\lambda_{1}^{2}=16$,
and $\left\Vert \tilde{\mathbf{b}}_{2}\right\Vert ^{2}=\lambda_{2}^{2}=28$.
On the contrary, if we turn $\mathbf{B}$ into a real basis and perform
real LLL (whose Lovasz's parameter is $1$) on it, the square norm
of the reduced vectors are respectively $16$, $16$, $31$, and $28$.
In its reduced basis, the first two vectors are not independent over
$K$, and the second shortest vector is in the last position. In this
scenario only the Minkowski reduction on the real basis can have the
same effect as our algebraic lattice reduction, whose reduced vectors
respectively have square norms $16$, $16$, $28$, and $28$.

\noindent \textbf{Example 2 (non-Euclidean domain).} Consider the
field $K=\mathbb{Q}(\sqrt{-5})$ and its maximal ring of integers
$\mathbb{Z}[\sqrt{-5}]$. By Proposition 1, this field is an example
of a non-norm Euclidean field. We begin with the following basis:
\[
\mathbf{B}=\left[\begin{array}{cc}
2+3\sqrt{-5} & 8+\sqrt{-5}\\
2+\sqrt{-5} & 2
\end{array}\right].
\]
Performing algebraic reduction on this basis consists of a single
size reduction, resulting in the basis 
\[
\tilde{\mathbf{B}}=\left[\begin{array}{cc}
2+3\sqrt{-5} & 6-2\sqrt{-5}\\
2+\sqrt{-5} & -\sqrt{-5}
\end{array}\right].
\]

Such a basis is reduced in the sense of Gauss whose vectors have square
lengths of $58$ and $61$. However, running real LLL over the corresponding
four dimensional basis returns reduced vectors with respective square
lengths $20,30,26,39$. As such, we conclude that the algebraic Gauss's
algorithm does not guarantee an output that corresponds to the successive
minima of the lattice if the chosen field is not Euclidean.

\section{Closing remarks}

In this paper, we have shown that it is possible to successfully build
a polynomial time algorithm that returns a basis that corresponds to
the successive minima. However, we have not addressed the lattice
reduction problem for non-Euclidean imaginary quadratic domains. Indeed,
all the quadratic forms listed in this paper become non-positive definite
when the respective field is not Euclidean (this can be easily seen
by example), however this does not immediately imply that reduction
fails over these fields. In our second numerical example, we have
shown that our definition of Gauss reduction, although converges to
a ``reduced'' basis, returns a basis that is much larger than the
actual successive minima of the lattice. The first question that could
be addressed in further research is whether the algorithm is optimal
for any lattices spanned over non-Euclidean imaginary quadratic rings,
and if not, is it guaranteed that there exists a unimodular transformation
that maps any basis to a new basis that corresponds to the successive
minima of the lattice. In the event that the answer to the first question
is negative and the second is positive, does there exist a modified
algorithm (possibly also polynomial time) that is optimal over lattices
over the said non-Euclidean ring? Another area to explore is reduction
over ``trace-Euclidean'' domains, i.e. domains where, for all $x\in K$,
there exists a $q\in\mathcal{O}_{K}$ such that $|\trace_{K/\mathbb{Q}}((x-q)\overline{(x-q)})|<1$
(for imaginary quadratic fields, trace-Euclideanity is equivalent
to norm-Euclideanity). 



\begin{thebibliography}{10}
\providecommand{\url}[1]{#1}
\csname url@samestyle\endcsname
\providecommand{\newblock}{\relax}
\providecommand{\bibinfo}[2]{#2}
\providecommand{\BIBentrySTDinterwordspacing}{\spaceskip=0pt\relax}
\providecommand{\BIBentryALTinterwordstretchfactor}{4}
\providecommand{\BIBentryALTinterwordspacing}{\spaceskip=\fontdimen2\font plus
\BIBentryALTinterwordstretchfactor\fontdimen3\font minus
  \fontdimen4\font\relax}
\providecommand{\BIBforeignlanguage}[2]{{%
\expandafter\ifx\csname l@#1\endcsname\relax
\typeout{** WARNING: IEEEtran.bst: No hyphenation pattern has been}%
\typeout{** loaded for the language `#1'. Using the pattern for}%
\typeout{** the default language instead.}%
\else
\language=\csname l@#1\endcsname
\fi
#2}}
\providecommand{\BIBdecl}{\relax}
\BIBdecl

\bibitem{Tunali2015}
N.~E. Tunali, Y.~Huang, J.~J. Boutros, and K.~R. Narayanan, ``Lattices over
  {E}isenstein integers for compute-and-forward,'' \emph{{IEEE} Trans. Inf.
  Theory}, vol.~61, no.~10, pp. 5306--5321, 10 2015.

\bibitem{DBLP:journals/tit/HuangNW18}
\BIBentryALTinterwordspacing
Y.~Huang, K.~R. Narayanan, and P.~Wang, ``Lattices over algebraic integers with
  an application to compute-and-forward,'' \emph{{IEEE} Trans. Information
  Theory}, vol.~64, no.~10, pp. 6863--6877, 2018.
\BIBentrySTDinterwordspacing

\bibitem{Vazquez-CastroO14}
M.~A.~V. Castro and F.~E. Oggier, ``Lattice network coding over euclidean
  domains,'' in \emph{22nd European Signal Processing Conference, {EUSIPCO}
  2014, Lisbon, Portugal, September 1-5, 2014}, 2014, pp. 1148--1152.

\bibitem{Kim2017}
T.~Kim and C.~Lee, ``Lattice reductions over euclidean rings with applications
  to cryptanalysis,'' in \emph{Proc. Cryptography and Coding - 16th {IMA}
  International Conference, {IMACC}, Oxford, UK, 2017}, ser. Lecture Notes in
  Computer Science, vol. 10655.\hskip 1em plus 0.5em minus 0.4em\relax
  Springer, 2017, pp. 371--391.

\bibitem{AlbrechtBD16}
M.~R. Albrecht, S.~Bai, and L.~Ducas, ``A subfield lattice attack on
  overstretched {NTRU} assumptions - cryptanalysis of some {FHE} and graded
  encoding schemes,'' in \emph{Advances in Cryptology - {CRYPTO} 2016 - 36th
  Annual International Cryptology Conference, Santa Barbara, CA, USA}, 2016,
  pp. 153--178.

\bibitem{Napias1996}
H.~Napias, ``A generalization of the {LLL}-algorithm over euclidean rings or
  orders,'' \emph{Journal de Th{\'{e}}orie des Nombres de Bordeaux}, vol.~8,
  no.~2, pp. 387--396, 1996.

\bibitem{Fieker1996}
C.~Fieker and M.~Pohst, ``On lattices over number fields,'' in
  \emph{Algorithmic Number Theory, Second International Symposium, ANTS-II,
  Talence, France, May 18-23, 1996, Proceedings}, ser. Lecture Notes in
  Computer Science, vol. 1122.\hskip 1em plus 0.5em minus 0.4em\relax Springer,
  1996, pp. 133--139.

\bibitem{Fieker2010}
C.~Fieker and D.~Stehl{\'{e}}, ``Short bases of lattices over number fields,''
  in \emph{Proc. Algorithmic Number Theory, 9th International Symposium,
  ANTS-IX, Nancy, France}, ser. Lecture Notes in Computer Science, vol.
  6197.\hskip 1em plus 0.5em minus 0.4em\relax Springer, 2010, pp. 157--173.

\bibitem{DBLP:lyu2018-ALLL}
\BIBentryALTinterwordspacing
S.~Lyu, C.~Porter, and C.~Ling, ``Performance limits of lattice reduction over
  imaginary quadratic fields with applications to compute-and-forward,'' in
  \emph{{IEEE} Information Theory Workshop, {ITW} 2018, Guangzhou, China},
  2018.
\BIBentrySTDinterwordspacing

\bibitem{Micciancio2002}
D.~Micciancio and S.~Goldwasser, \emph{Complexity of Lattice Problems}.\hskip
  1em plus 0.5em minus 0.4em\relax Boston, MA: Springer, 2002.

\bibitem{Yao2002}
H.~Yao and G.~W. Wornell, ``Lattice-reduction-aided detectors for {MIMO}
  communication systems,'' in \emph{Proc. Global Telecommunications Conference
  ({GLOBECOM}), Taipei, Taiwan, 2002}.\hskip 1em plus 0.5em minus 0.4em\relax
  {IEEE}, 2002, pp. 424--428.

\end{thebibliography}
\end{document}